\documentclass[reqno]{amsart}
\usepackage{amsmath,amsfonts,amscd,amssymb}
\usepackage{mathtools}
\usepackage{dsfont}
\usepackage{lineno}
\usepackage{multirow}
\usepackage{tikz}
\usepackage{float}
\usetikzlibrary{decorations.pathreplacing}
\usetikzlibrary{backgrounds} 
\usepackage{placeins}
\usepackage[mathscr]{euscript}
\usepackage[small]{caption}
\usepackage{graphicx}
\usepackage{stmaryrd}
\usepackage{color}
\usepackage{changes}

\DeclareMathOperator*{\argminA}{arg\,min}

\makeatletter
\@addtoreset{equation}{section}
\makeatother

\newtheorem{theorem}{Theorem}[section]

\newtheorem{proposition}[theorem]{Proposition}
\newtheorem{corollary}[theorem]{Corollary}

\newtheorem{remark}[theorem]{Remark}

\newcounter{as}[section]

\definecolor{bblue}{rgb}{.2,0.2,.8}
\usepackage[foot]{amsaddr}

\begin{document}

\title[Parameter estimation in dynamical systems via Statistical Learning]{Parameter estimation in dynamical systems via Statistical Learning: a reinterpretation of Approximate Bayesian Computation applied to COVID-19 spread\thanks{Institute of Mathematics and Statistics, Universidade de S\~{a}o
		Paulo, R. do Mat\~{a}o, 1010 - Butant\~{a}, S\~{a}o Paulo - SP,
		05508-090, Brazil. E-mail: \texttt{dmarcondes@ime.usp.br}}}
\author{Diego Marcondes}
\begin{abstract}
We propose a robust parameter estimation method for dynamical systems based on Statistical Learning techniques which aims to estimate a set of parameters that well fit the dynamics in order to obtain robust evidences about the qualitative behaviour of its trajectory. The method is quite general and flexible, since it does not rely on any specific property of the dynamical system, and represents a reinterpretation of Approximate Bayesian Computation methods through the lens of Statistical Learning. The method is specially useful for estimating parameters in epidemiological compartmental models in order to obtain qualitative properties of a disease evolution. We apply it to simulated and real data about COVID-19 spread in the US in order to evaluate qualitatively its evolution over time, showing how one may assess the effectiveness of measures implemented to slow the spread and some qualitative features of the disease current and future evolution.
\\

\end{abstract}

\address{Institute of Mathematics and Statistics, Universidade de S\~{a}o
Paulo, R. do Mat\~{a}o, 1010 - Butant\~{a}, S\~{a}o Paulo - SP,
05508-090, Brazil. \\
e-mail: \texttt{dmarcondes@ime.usp.br}}
\maketitle

\section{Introduction}
\label{intro}

Dynamical systems are important tools in Applied Mathematics for modelling phenomena studied by many branches of science, such as physics, cosmology, biology, epidemiology, medicine, chemistry and engineering \cite{coley2013,jackson2015,ma2009,rosen1970,strogatz2001}, being applied to describe the deterministic evolution in time of certain processes, when starting from a fixed initial condition. The equations which govern this deterministic evolution are specific to each application, and depend on some unknown parameters related to the phenomena being modelled. For example, dynamical systems applied to physics and cosmology usually depend on universal constants, while compartmental models in epidemiology depend on rates related to the spread of infectious diseases, as their basic reproductive number \cite{capasso2008,vynnycky2010}. Hence, in order to model phenomena by dynamical systems, one should find parameters which \textit{represent} reality.

Some parameters, such as universal constants, are known, while others can be measured, but some are, at principle, completely unknown and need to be estimated. Such estimation is carried out from available data on the evolution of the process in a time interval, when one observes a partial trajectory of a system and try to determine the parameters which generate an evolution which fit the observed one. This problem may not have an unique solution, since the evolution in an interval may not be enough to identify an unique vector of parameters, and, even when the parameters are identifiable, there may not exist a computable algorithm capable of finding them (see \cite{canto2017,gu2013} for a discussion on identifiability of parameters).

In this scenario, there are three features usually present in methods for parameter estimation in dynamical systems: closed formulae which relates the parameters to observed values, as is the case of the basic reproduction number in epidemiological models \cite{wallinga2007}; an algorithm which randomly selects parameters from a set of candidates; and some measure of goodness-of-fit which evaluates if the evolution generated by each parameter well fit the observed one. There are countless methods for parameter estimation in dynamical systems which present these, and other features. We refer, for instance, \cite{ghahramani1996parameter,ghahramani1999learning,green2015,raue2015data2dynamics,xu2015application} and the references therein.

Among these methods, there is the class known as Approximate Bayesian Computation (ABC) which has been developed aiming to circumvent the calculation of likelihood functions when inferring the posterior distribution of parameters, what is usually computationally unfeasible or too costly to perform. In the simplest form of ABC, the ABC rejection sampler \cite{pritchard1999}, parameters $\theta$ are simulated from the prior distribution $\pi(\theta)$ and a system evolution $x$, sampled from $f(x|\theta)$, is compared with the observed one $x^{\star}$. If the evolutions are \textit{close enough}, i.e., $d(x,x^{\star}) < \epsilon$, the sampled parameter is accepted; otherwise it is rejected. The accepted parameters will be a sample of the distribution $\pi(\theta|d(x,x^{\star}) < \epsilon)$ which is an approximation of the posterior $\pi(\theta|x)$. 

Since the rejection rate of such an algorithm may be too great, demanding a very large sample of $\pi(\theta)$, more refined versions of this method, with higher acceptance rate, such as ABC Markov Chain Monte Carlo (MCMC) \cite{marjoram2003} and ABC Sequential Monte Carlo (SMC) \cite{toni2009}, have been proposed. These versions, and others in the literature (see \cite{beaumont2010,csillery2010,lintusaari2017,marin2012,sunnaaker2013} for a through review of ABC methods), follows more or less the same idea presented above, of sampling parameters from a prior distribution and rejecting them as samples of an approximation of the posterior if data generated by them does not \textit{well fit} observed data, and the main difference between the methods is how the prior distribution is sampled.

In this paper we propose a robust approach for parameter estimation in dynamical systems based on Statistical Learning techniques which is actually an ABC algorithm, but developed from another point of view. The method does not look for the vector of parameters which \textit{best} fit the observed evolution, but rather seeks to find a set of parameters which \textit{well fit}, in some sense, such evolution. This method is suitable when one wants to identify qualitative properties of the observed dynamics, rather than know it exactly, or when the evolution is only approximated by a dynamical system. Furthermore, since the dynamics may be sensitive to the choice of parameters, considering a set of parameters may mitigate such sensitivity, since properties satisfied by all of them may also be satisfied by the observed dynamics (see for example the role of parameter $s$ in \cite{peixoto2020}). Hence, considering a set of parameters and looking for common properties shared by the evolutions they generate account for the robustness of the method.

In the proposed approach, one chooses a set of candidate parameters, based on all prior knowledge about the modelled dynamics, and fix a binary rule which, given a candidate parameter, determine if its evolution well fit, or not, the observed dynamics. We say that a parameter is \textit{good} if its evolution well fit the dynamics, and is \textit{bad} otherwise. The method aims to find good parameters among candidate ones by sampling and testing. To this end, we apply Statistical Learning techniques to determine an upper bound for the sample size needed to find a proportion at least $1 - c, 0 < c < 1$, of the good parameters with high probability, which is a measure of the computational complexity of the approach.

The method is quite general, since it does not rely on any specific property of the dynamical system. Observe that the agenda the method seeks to carry out is much more humble than trying to identify exactly the value of the optimal parameter, since its aim is only to find good parameters in a pre-defined set. This is more plausible to be achieved in practice, specially in processes which are only approximated by dynamical systems or whose parameters change from time to time. Furthermore, since the method is an instance of ABC, much of the theory and algorithms developed to it may be employed to carry out the qualitative analysis of the proposed framework.

The Statistical Learning approach is specially suitable when modelling disease spread by compartmental epidemiological models since they are completely defined by a handful of real valued parameters and there is a great interest in qualitative properties of their trajectory, rather than knowing it exactly. We illustrate in an example with real data about the spread of COVID-19 in the US how this method can be quite useful in obtaining qualitative information about the evolution of the process. Indeed, we will see how the parameters of a SEIR model change from time to time due to the enforcement of measures to slow the disease spread, demanding a constant fit of them which in itself can provide evidences about the disease spread and the effectiveness of measures to overcome it. 

Although there is a vast application of ABC for compartmental epidemiological models \cite{brown2018}, and for COVID-19 spread in special \cite{arenas2020,de2020emergence}, we perform what we believe is a novel application of it, in which the parameters of the model are estimated weekly and their behaviour over time presents invaluable information about the effectiveness of measures implemented to slow the spread, and some qualitative features of the disease current and future evolution.

In Section \ref{Sec2} we discuss the problem of robustly estimating parameters in dynamical systems from their evolution in an interval. In Section \ref{Sec3} we propose a method based on Statistical Learning techniques to perform such robust estimation. In Section \ref{Sec4} we present a couple of examples of the proposed method for a SIR model on a simulated dataset, and for a SEIR model for the spread of COVID-19 in the US. In Section \ref{Sec5} we give our final remarks. 

\section{Robust parameter estimation in dynamical systems}
\label{Sec2}

Let $\{\{X_{\theta}(t)\}_{t \in \mathbb{N}}:\theta \in \Theta \}$ be a family of discrete time dynamical systems, with a same metric phase space $(\Omega,d)$, indexed by parameters $\theta \in \Theta$. The parametric space $\Theta$ may be such that $\Theta \subseteq \mathbb{R}^{d}, d \geq 1$, when the parameters are time independent, or $\Theta \subseteq \mathbb{N} \times \mathbb{R}^{d}$, when the parameters are time dependent. We assume that the initial condition is the same for all systems in the family: $X_{\theta}(0) = x_{0} \in \Omega$ for all $\theta \in \Theta$.

The problem of parameter estimation in dynamical systems we consider in this paper is characterized when one observes the time evolution of a process $X_{\theta^{\star}}(t)$ for $t \in \{1,\dots,T\}$, with a fixed $T \geq 1$, and wants to identify the parameter $\theta^{\star} \in \Theta$ which generated such evolution. On the one hand, for $T \in \mathbb{N}$, the map
\begin{linenomath}
	\begin{align*}
	&\phi_{T}: \Theta \to \Omega^{T} \\
	& \phi_{T}(\theta) = \{X_{\theta}(t)\}_{t=1}^{T} 
	\end{align*} 
\end{linenomath}
which maps each parameter $\theta \in \Theta$ to the evolution of $X_{\theta}(t)$ until $t = T$, is in general not invertible, i.e., $\theta$ is not identifiable, so the evolution until a time $T$ does not define uniquely the system. Furthermore, the set $\phi^{-1}_{T}(\{X_{\theta^{\star}}(t)\}_{t=1}^{T})$ may not be computable in practical problems, so one cannot identify a subset of candidate parameters by inverting such a map. On the other hand, it is not computationally feasible to test all parameters in $\Theta$, in case it has infinite elements, or to perform an efficient grid search of $\Theta$, when it is multidimensional, to find candidate parameters.

Precisely estimating parameters of a dynamical system from the time evolution in an interval may have an intrinsic lack of robustness due to the sensitivity of the evolution on parameter $\theta$ (see \cite{wu2013} for a review of sensitivity analysis in compartmental models). This implies that, even if we estimate $\theta^{\star}$ by a $\hat{\theta}$ \textit{close} to $\theta^{\star}$, $d\left(X_{\theta^{\star}}(t), X_{\hat{\theta}}(t)\right)$ may be too great for $t > T$, rendering the estimative useless for the problem at hand, specially when the aim of estimating parameters of a dynamical system is predicting exactly its trajectory.  

A manner of increasing the robustness of the estimation method is, rather than estimating $\theta^{\star}$ precisely by a $\hat{\theta}$, to identify a set $\hat{\Theta} \subseteq \Theta$ of possible values for $\theta^{\star}$ which will generate evolutions of the system: $\{\{X_{\theta}(t)\}_{t \in \mathbb{N}}: \theta \in \hat{\Theta}\}$. Establishing common properties of these evolutions, one may obtain evidences about the behaviour of $X_{\theta^{\star}}(t)$ for $t > T$. Even though the chosen models are still sensible to the parameters, there may be some qualitative behaviour, common to all of them, which may also be shared with the evolution generated by $\theta^{\star}$, so this method adds to one understanding of the system trajectory.

The chosen models should \textit{well fit} the observed evolution until time $T$. The definition of \textit{well fit} is problem dependent and shall be given by a fitness map
\begin{linenomath}
	\begin{equation*}
	F: \Theta \times \Omega^{T} \to \{0,1\}
	\end{equation*}
\end{linenomath}
which, for a $\theta \in \Theta$ and observed evolution in $\Omega^{T}$, attributes $1$ if the system generated by $\theta$ \textit{well fit} the observed evolution, and attributes $0$ otherwise. An example of fitness map is

\begin{linenomath}
	\begin{align}
	\label{example_FM}
	F(\theta,\{X_{\theta^{\star}}(t)\}_{t=1}^{T}) = \min\limits_{t \in \{1,\dots,T\}}  \mathds{1}\Big\{d(X_{\theta}(t),X_{\theta^{\star}}(t)) \leq \delta(t,X_{\theta^{\star}}(t))\Big\}
	\end{align}
\end{linenomath}
for $\delta(t,X_{\theta^{\star}}(t)) > 0$, dependent on $t$ and $X_{\theta^{\star}}(t)$. By \eqref{example_FM}, a model well fit $\{X_{\theta^{\star}}\}_{t=1}^{T}$ if its time evolution is relatively close to the observed one for all $t \leq T$. Another example of fitness map is
\begin{linenomath}
	\begin{align}
	\label{example_FM2}
	F(\theta,\{X_{\theta^{\star}}(t)\}_{t=1}^{T}) = \mathds{1}\Big\{\frac{1}{T} \sum_{t=1}^{T} d(X_{\theta}(t),X_{\theta^{\star}}(t)) \leq \delta\Big\}
	\end{align}
\end{linenomath}
for a constant $\delta > 0$, so a model well fit the trajectory if its mean distance to it over all times until $T$ is lesser than $\delta$. The fitness map is problem dependent and should be chosen according to the purpose of the dynamical system.

The estimation method for dynamical systems parameters proposed in this paper aims to find a subset $\hat{\Theta} \subseteq \Theta$ of parameters such that $F(\theta,\{X_{\theta^{\star}}(t)\}_{t=1}^{T}) = 1$ for all $\theta \in \hat{\Theta}$. As an alternative to invert map $\phi_{T}$ or perform a grid search on $\Theta$, we will propose a probabilistic consistent method to randomly select parameters from $\Theta$ and take $\hat{\Theta}$ as the ones with goodness-of-fit according to a fitness map $F$. This method is based on Statistical Learning techniques \cite{devroye1996,vapnik1998,vapnik2000} and is a reinterpretation of ABC methods.

\begin{remark}
	In order to develop the method we assume there is no noise in the system, what is equivalent to the existence of a $\theta^{\star} \in \Theta$ which generated the observed data $X_{\theta^{\star}}(t)$. However, the framework may also be applied to cases where there is noise and the observed dynamics $X(t)$ is not exactly equal to $X_{\theta^{\star}}$ for a $\theta^{\star} \in \Theta$.
\end{remark}

\begin{remark}
	Goodness-of-fit, as defined above, is a dichotomous concept: a model either does, or does not, fit the observed evolution. Even if we considered some likelihood function or distance between the observed and generated evolution as numerical measures of goodness-of-fit, we would have to choose a threshold for it in order to obtain $\hat{\Theta}$. Hence, the class of fitness maps contemplate these numerical goodness-of-fit measures, since they may be composed by such measures, as in example \eqref{example_FM}. Therefore, dichotomous goodness-of-fit are enough for our purposes.
\end{remark}

\begin{remark}
	The development of ABC is based on the avoidance to calculate likelihood functions when inferring posterior distributions, while to develop this method we depart from candidate parameters and want to find models which well fit data. Although with slightly different purposes, ABC and the proposed method are analogous and may return the same result if performed in consonance (cf. Section \ref{relABC}).
\end{remark}

\section{Parameter estimation via Statistical Learning}
\label{Sec3}

\subsection{Statistical Learning}

The Statistical Learning method under the Empirical Risk Minimization (ERM) paradigm, restricted to classification problems, may be stated as follows. Let $Z$ be a random vector, and $Y$ a random variable, defined on a same probability space $(\Lambda,\mathcal{S},\mathbb{P})$, with ranges $\mathcal{Z} \subset \mathbb{R}^{d}, d \geq 1,$ and $\{0,1\}$, respectively. Denote $P(z,y) \coloneqq \mathbb{P}(Z \leq z,Y \leq y)$ as the joint probability distribution of $(Z,Y)$ at point $(z,y) \in \mathcal{Z} \times \{0,1\}$, which we assume unknown, but fixed. Define a sample $\mathcal{D}_{N} = \{(Z_{1},Y_{1}), \dots, (Z_{N},Y_{N})\}$ as a sequence of independent and identically distributed random vectors, defined on $(\Lambda,\mathcal{S},\mathbb{P})$, with joint distribution $P$.

Let $\mathcal{H}$ be a set of functions with domain $\mathcal{Z}$ and image $\{0,1\}$, whose typical element we denote by $h: \mathcal{Z} \to \{0,1\}$. We call $\mathcal{H}$ Hypotheses Space. For each hypothesis $h$ in $\mathcal{H}$, we assign a value indicating the error incurred by the use of such hypothesis to predict $Y$ from the value of $Z$, i.e., how well $h(Z)$ predicts $Y$. To obtain such an error measure, we consider the simple loss function $\ell: \mathcal{Z} \times \{0,1\} \times \mathcal{H} \mapsto \mathbb{R}$ given by
\begin{linenomath}
	\begin{equation*}
	l(z,y;h) = \begin{cases}
	1, & \text{ if } h(z) \neq y\\
	0, & \text{ if } h(z) = y
	\end{cases}
	\end{equation*}
\end{linenomath}
for $(z,y,h) \in \mathcal{Z} \times \{0,1\} \times \mathcal{H}$. In this framework, the loss of predicting $y$ by $h(z)$ is either zero if $h(z) = y$, or one if $h(z) \neq y$. The prediction error, known  in the literature as out-of-sample error, risk or loss \cite{abu2012,devroye1996,vapnik1998,vapnik2000}, of a hypothesis $h \in \mathcal{H}$ is defined as
\begin{linenomath}
	\begin{equation*}
	L(h) \coloneqq \mathbb{E}[\ell(Z,Y;h)] =  \mathbb{P}\left(h(Z) \neq Y\right)
	\end{equation*}
\end{linenomath}
in which $\mathbb{E}$ means expectation under $\mathbb{P}$. This is the probability of error incurred when hypothesis $h$ is used to predict $Y$ from $Z$.

The out-of-sample error is fixed, but unknown, as is $P$. Therefore, in order to asses the out-of-sample error of a hypothesis, one needs to estimate it. An estimator for $L$ may be obtained by the empirical error on sample $\mathcal{D}_{N}$, called in-sample error and defined as
\begin{linenomath}
	\begin{equation*}
	L_{\mathcal{D}_{N}}(h) \coloneqq \frac{1}{N} \sum_{i=1}^{N} \mathds{1}\left\{h(Z_{i}) \neq Y_{i}\right\}
	\end{equation*}
\end{linenomath}
for $h \in \mathcal{H}$. This is simply the classification error of $h$ when applied to classify the points in sample $\mathcal{D}_{N}$.

The main goal of Statistical Learning in this context is to approximate target hypotheses, which are minimizers of the out-of-sample error in $\mathcal{H}$. These hypotheses are in set
\begin{linenomath}
	\begin{equation*}
	h^{\star} \coloneqq \argminA\limits_{h \in \mathcal{H}} L(h)
	\end{equation*}
\end{linenomath}
and satisfy $L(h^{\star}) \leq L(h), \forall h \in \mathcal{H}$. We assume throughout this paper that $L(h^{\star}) = 0$, i.e., there is a hypothesis in $\mathcal{H}$ which predicts $Y$ from the values of $Z$ with probability one. Under the ERM principle, which proposes the minimization of the in-sample error as a method to approximate target hypotheses, we estimate them by 
\begin{linenomath}
	\begin{align*}
	\hat{h} \coloneqq \argminA\limits_{h \in \mathcal{H}} L_{\mathcal{D}_{N}}(h),
	\end{align*}	
\end{linenomath}
the hypotheses which minimize the classification error in sample $\mathcal{D}_{N}$. Observe that $\hat{h} = \hat{h}_{\mathcal{D}_{N}}$ is actually dependent on $\mathcal{D}_{N}$, but we drop the subscript to ease notation. 

In the ERM paradigm we are interested in considering Hypotheses Spaces $\mathcal{H}$ which are \textit{Probably Approximately Correct} (PAC)-Learnable \cite{valiant1984}. We say that Hypotheses Space $\mathcal{H}$ is PAC-Learnable\footnote{See \cite[Definition~3.3]{shalev2014} for a more general definition of (Agnostic) PAC-learnability.} if
\begin{linenomath}
	\begin{equation}
	\label{PAC_lim}
	\lim\limits_{N \rightarrow \infty} \mathbb{P}\Big(L(\hat{h}) > \epsilon\Big) = 0
	\end{equation} 
\end{linenomath}
for all $\epsilon > 0$. This means that, if the sample size $N$ tends to infinity, the out-of-sample error of the estimated hypotheses $\hat{h}$ is arbitrarily close to the out-of-sample error of the target hypotheses, i.e., zero, with high probability. In Probability Theory \cite{billingsley2008probability}, it means that $L(\hat{h})$ converges in probability to zero. In a PAC-Learnable Hypotheses Space we may estimate $h^{\star}$ arbitrarily well, with great confidence, if the sample size is \textit{great enough}.

In this paper, we will be interested in Hypotheses Spaces with a finite number of elements, i.e., $|\mathcal{H}| < \infty$. Such Hypotheses Spaces are not only PAC-Learnable, but we can establish a \textit{distribution-free} rate of convergence to zero of limit \eqref{PAC_lim}. Distribution-free in this context means that such a rate is true for any possible joint distribution $P$ of $(Z,Y)$. Theorem \ref{bound_error} presents such convergence rate. Its elementary proof, deduced for the first time by \cite{vapnik1974theory} and presented in \cite[Theorem~12.1]{devroye1996}, is in the Appendix.

\begin{theorem}
	\label{bound_error}
	Assume that $|\mathcal{H}| < \infty$ and $L(h^{\star}) = 0$. Then, for every $N$ and $\epsilon > 0$,
	\begin{equation}
	\label{bound1}
	\mathbb{P}\left(L(\hat{h}) > \epsilon\right) \leq |\mathcal{H}|e^{-N\epsilon}.
	\end{equation}	
\end{theorem}

We obtain as a corollary of Theorem \ref{bound_error} an upper bound for the sample size needed to estimate a hypotheses $\hat{h}$ such that $L(\hat{h}) < \epsilon$ with probability at least $1 - \delta$. This bound is obtained by solving $|\mathcal{H}| e^{-N\epsilon} = \delta$ for $N$, fixed $\epsilon$ and $\delta$, and is stated below.

\begin{corollary}
	\label{sample_size}
	Define $m_{\mathcal{H}}:(0,1)^{2} \rightarrow \mathbb{Z}_{+}$ as 
	\begin{linenomath}
		\begin{equation*}
		m_{\mathcal{H}}(\epsilon,\delta) =\frac{1}{\epsilon} \log \frac{|\mathcal{H}|}{\delta}.
		\end{equation*}
	\end{linenomath}
	For all $\epsilon,\delta \in (0,1)$, if $N \geq m_{\mathcal{H}}(\epsilon,\delta)$, then
	\begin{linenomath}
		\begin{equation}
		\label{cond_samplesize}
		\mathbb{P}\Big(L(\hat{h}) < \epsilon\Big) \geq 1 - \delta.
		\end{equation} 
	\end{linenomath}	
\end{corollary}

\begin{remark}
	The concept of PAC-Learnability is related to Vapnik-Chervonenkis (VC) Theory \cite{vapnik1998,vapnik2000,vapnik1974ordered,vapnik1971uniform,vapnik1974theory,vapnik1974ordered2}. Bounds analogous to \eqref{bound1} may be obtained from VC Theory for Hypotheses Spaces with infinite cardinality or such that $L(h^{\star}) > 0$. For more details see \cite{devroye1996,shalev2014,vapnik1998,vapnik2000}.
\end{remark}

\subsection{Parameter estimation}

Robust parameter estimation in dynamical systems may be achieved via Statistical Learning with the framework above. Let $\mathcal{Z} \subseteq \Theta$ be a finite subset of candidate parameters. It can be, for example, a grid of $\Theta \subseteq \mathbb{R}^{d}$. The joint distribution of $(Z,Y)$ is such that, for $\theta \in \mathcal{Z}$ and $y \in \{0,1\}$,
\begin{linenomath}
	\begin{equation*}
	\mathbb{P}\left(Z = \theta,Y = y\right) = q(\theta) \mathds{1}\left\{F(\theta,\{X_{\theta^{\star}}(t)\}_{t=1}^{T}) = y\right\}
	\end{equation*}
\end{linenomath}
in which $0 < q(\theta) < 1$ and $\sum_{\theta \in \mathcal{Z}} q(\theta) = 1$, i.e., $Z$ has a discrete probability distribution $q$ in $\mathcal{Z}$, and $Y$ is equal to $F(Z,\{X_{\theta^{\star}}(t)\}_{t=1}^{T})$ with probability one. Hence, $Z$ represents a parameter chosen randomly from $\mathcal{Z}$ according to probability distribution $q$ and $Y$ is equal to one if the evolution generated by this parameter well fit the observed evolution, and zero otherwise. Therefore, independently of the Hypotheses Spaces $\mathcal{H}$, the loss function is given by
\begin{linenomath}
	\begin{equation*}
	\ell(\theta,y;h) = \begin{cases}
	1, \text{ if } h(\theta) \neq F(\theta,\{X_{\theta^{\star}}\}_{t = 1}^{T})\\
	0, \text{ if } h(\theta) = F(\theta,\{X_{\theta^{\star}}\}_{t = 1}^{T})
	\end{cases}.
	\end{equation*}
\end{linenomath}

The loss function tells us what the hypothesis $h$ should do: point out what are the good (well fit) and what are the bad parameters in $\mathcal{Z}$. Let $p \in \mathbb{N}$ be the number of good parameters in $\mathcal{Z}$:
\begin{linenomath}
	\begin{equation*}
	p \coloneqq \left|\{\theta \in \mathcal{Z}: F(\theta,\{X_{\theta^{\star}}\}_{t = 1}^{T}) = 1\}\right|.
	\end{equation*}
\end{linenomath}
The Hypotheses Space considered in this scenario is
\begin{linenomath}
	\begin{equation*}
	\mathcal{H} \coloneqq \left\{h:\mathcal{Z} \to \{0,1\}: h(\theta) = 0 \text{ for all } \theta \text{ s.t. } F(\theta,\{X_{\theta^{\star}}\}_{t = 1}^{T}) = 0\right\}
	\end{equation*}
\end{linenomath}
so a hypothesis $h$ should be interpreted as the characteristic function of a subset of $\{\theta \in \mathcal{Z}: F(\theta,\{X_{\theta^{\star}}\}_{t = 1}^{T}) = 1\}$. Indeed, if $\theta$ is such that $h(\theta) = 1$, then it is a good parameter, hence $\{\theta: h(\theta) = 1\}$ is a subset of the good parameters, for all $h \in \mathcal{H}$. In this scenario, $h^{\star}$ is the subset of the $p$ good parameters in $\mathcal{Z}$ and $L(h^{\star}) = 0$. Observe that
\begin{linenomath}
	\begin{equation*}
	|\mathcal{H}| = 2^{p},
	\end{equation*}
\end{linenomath}
so $\mathcal{H}$ has finite cardinality. 

The number $p$ of good parameters in $\mathcal{Z}$ is unknown, but strictly related to the choice of fitness map $F$. On the one hand, if the concept of \textit{well fit} is too restrict, then less parameters will be good, hence one has a \textit{small} $p$. In this case, one may not be able to obtain robust evidences about $X_{\theta^{\star}}(t)$ due to the small number of candidate evolutions. On the other hand, if the concept of \textit{well fit} is too loose, then more parameters may meet the criteria of a good one, so greater is $p$. As we want a robust estimation, $p$ should be much greater than one, so there are a handful of evolutions available from which to obtain evidences about the evolution of $X_{\theta^{\star}}(t)$. Nevertheless, if $p$ is too great, then probably our definition of well fit is too loose, so the candidate systems may not be homogeneously related to $X_{\theta^{\star}}$, rendering useless evidences about $X_{\theta^{\star}}(t)$. Hence, both $\mathcal{Z}$ and $F$ should be chosen based on all prior information available about the problem at hand, so $p \gg 1$, but the good parameters actually generate meaningful trajectories.

After one fixes $\mathcal{Z}$ and $F$, the ERM principle is then applied to estimate a set of at most $p$ good parameters. The estimated sets are
\begin{linenomath}
	\begin{equation*}
	\hat{h} = \Big\{h \in \mathcal{H}: h(Z_{i}) = F(Z_{i},\{X_{\theta^{\star}}\}_{t = 1}^{T}), \ \forall i \in \{1,\dots,N\}\Big\}
	\end{equation*}
\end{linenomath}
i.e., the subsets of good parameters which contain all of the observed in the sample, which generates the following set of good parameters:
\begin{linenomath}
	\begin{equation*}
	\hat{\Theta} = \Big\{Z_{i}: F(Z_{i},\{X_{\theta^{\star}}\}_{t=1}^{T}) = 1\Big\}.
	\end{equation*}
\end{linenomath}
Therefore, the set \textit{de facto} estimated by the ERM principle is $\hat{h}$ such that
\begin{linenomath}
	\begin{align*}
	\hat{h}(\theta) = \begin{cases}
	1, & \text{ if } \theta = Z_{i} \text{ s.t. } F(Z_{i},\{X_{\theta^{\star}}\}_{t = 1}^{T}) = 1 \\ 
	0, & \text{ if } \theta \in \mathcal{Z}\setminus\{Z_{i}: F(Z_{i},\{X_{\theta^{\star}}\}_{t = 1}^{T}) = 1\}
	\end{cases}
	\end{align*}
\end{linenomath}
that is, the set containing only the good parameters in the sample. From now on when we denote $\hat{h}$ we mean the set above.

Observe that
\begin{linenomath}
	\begin{align*}
	L(\hat{h}) &= \sum_{\substack{\theta \in \mathcal{Z}\\ F(\theta,\{X_{\theta^{\star}}\}_{t = 1}^{T}) = 1}} (1 - \hat{h}(\theta)) \ q(\theta) = \sum_{\substack{\theta \in \mathcal{Z}\\
			F(\theta,\{X_{\theta^{\star}}\}_{t = 1}^{T}) = 1\\
			\theta \notin \{Z_{1},\dots,Z_{N}\}}} \ \ q(\theta),
	\end{align*}
\end{linenomath}
is the measure of the good parameters not present in the sample. Since 
\begin{linenomath}
	\begin{equation*}
	G(\{X_{\theta^{\star}}\}_{t = 1}^{T}) \coloneqq \mathbb{P}(\theta: F(\theta,\{X_{\theta^{\star}}\}_{t = 1}^{T}) = 1),	
	\end{equation*}
\end{linenomath}
an upper bound for $L(\hat{h})$, may be very small, simply having $L(\hat{h}) < \epsilon$ is not meaningful if such probability is lesser than $\epsilon$. Hence, in order to assess if $L(\hat{h})$ is \textit{small}, one should compare it with $G(\{X_{\theta^{\star}}\}_{t = 1}^{T})$, to establish if it is \textit{small} relatively to $G(\{X_{\theta^{\star}}\}_{t = 1}^{T})$.

In order to perform such a comparison, we apply Theorem \ref{bound_error}, which states that the probability of having an out-of-sample error greater than $cG(\{X_{\theta^{\star}}\}_{t = 1}^{T})$, which means that a proportion at least $0 < c < 1$ of the measure of the set of all good parameters is not observed in the sample, is such that
\begin{linenomath}
	\begin{equation}
	\label{phat}
	\mathbb{P}\left(L(\hat{h}) > cG(\{X_{\theta^{\star}}\}_{t = 1}^{T})\right) \leq 2^{p} e^{-cG(\{X_{\theta^{\star}}\}_{t = 1}^{T})N}
	\end{equation}
\end{linenomath}
so it follows from Corollary \ref{sample_size} that, in order to obtain a proportion\footnote{In the sense of measure, not necessarily a proportion of the number of good parameters.} at least $(1-c)$ of the good parameters in the sample, with a probability greater than $1-\delta$, one needs a sample of size at least
\begin{linenomath}
	\begin{align}
	\label{psample_size}
	m_{\mathcal{H}}\left(c,\delta\right) &= \frac{1}{cG(\{X_{\theta^{\star}}\}_{t = 1}^{T})} \left[p \log 2 - \log \delta\right].
	\end{align}
\end{linenomath}
Assuming that $p \sim G(\{X_{\theta^{\star}}\}_{t = 1}^{T}) |\mathcal{Z}|$, the sample size above is meaningful only if
\begin{linenomath}
	\begin{align*}
	c \geq \log 2 - \frac{\log \delta}{p} > 0.69,
	\end{align*}
\end{linenomath}
for otherwise $m_{\mathcal{H}}\left(c,\delta\right) > |\mathcal{Z}|$. This results is reasonable since, as $N \to |\mathcal{Z}|$, the cardinality of $\{Z_{1},\dots,Z_{N}\}$ does not converge to $|\mathcal{Z}|$ with probability one for there will be points which are sampled more than once. Hence, to observe  a proportion of the good parameters one needs to sample more models than there are in $\mathcal{Z}$, so this method is better than an exhaustive search of $\mathcal{Z}$ only if one wants a small number of good parameters when compared to the number of good parameters in $\mathcal{Z}$.

The quantities on the right-hand side of \eqref{phat} and \eqref{psample_size} are not known, since $p$ and $G(\{X_{\theta^{\star}}\}_{t = 1}^{T})$ are unknown. However, the sample size above is an increasing function of $p$, therefore, given an upper bound to $p$, one can obtain a greater sample size that guarantees \eqref{cond_samplesize}. Also, $G(\{X_{\theta^{\star}}\}_{t = 1}^{T})$ may be estimated by a pre-sample as the proportion of good parameters in it.

If $q(\theta) = G(\{X_{\theta^{\star}}\}_{t = 1}^{T})/p$ for all $\theta$ such that $F(\theta,\{X_{\theta^{\star}}\}_{t = 1}^{T}) = 1$ then the bounds in \eqref{phat} and the sample size \eqref{psample_size} may be improved. The bound below is much better than \eqref{psample_size}, especially for $c$ close to $1$.

\begin{proposition}
	\label{prop_improve}
	If the measure $q$ conditioned on $F(\theta,\{X_{\theta^{\star}}\}_{t = 1}^{T}) = 1$ is uniform then
	\begin{linenomath}
		\begin{equation*}
		\mathbb{P}\left(L(\hat{h}) > cG(\{X_{\theta^{\star}}\}_{t = 1}^{T})\right) \leq \sum_{k=0}^{(1-c)p} \binom{p}{k} e^{-cG(\{X_{\theta^{\star}}\}_{t = 1}^{T})N}
		\end{equation*}
	\end{linenomath}
	from which follows
	\begin{linenomath}
		\begin{equation}
		\label{bound_m}
		m_{\mathcal{H}}\left(c,\delta\right) \leq \frac{1}{cG(\{X_{\theta^{\star}}\}_{t = 1}^{T})} \left[\log \binom{p}{(1-c)p} + \log\left(1 + \frac{(1-c)^{2} p^{2}}{cp + 1}\right) - \log \delta \right]
		\end{equation}
	\end{linenomath}
	if $c > 1/2$.
\end{proposition}

Figure \ref{m} presents $m_{\mathcal{H}}\left(c,\delta\right)$ given by \eqref{psample_size} (dashed) and its upper bound when $q$ conditioned on the good parameters is uniform (solid line), as a function of $c$, with $\delta = 0.1$, $G(\{X_{\theta^{\star}}\}_{t = 1}^{T}) = 0.001,$ $|\mathcal{Z}| = 10^{6}$ and $p = 1,000$. We see that, for example, to obtain $10\%$ of the good parameters one has to sample a proportion around $0.35$ and $0.8$ when $q$ conditioned is uniform and in general, respectively. Therefore, taking $q$ close to the uniform distribution when conditioned on the good parameters is better since less samples are needed to find a given proportion of the good parameters.

\begin{figure}[ht]
	\centering
	\includegraphics[width=\linewidth]{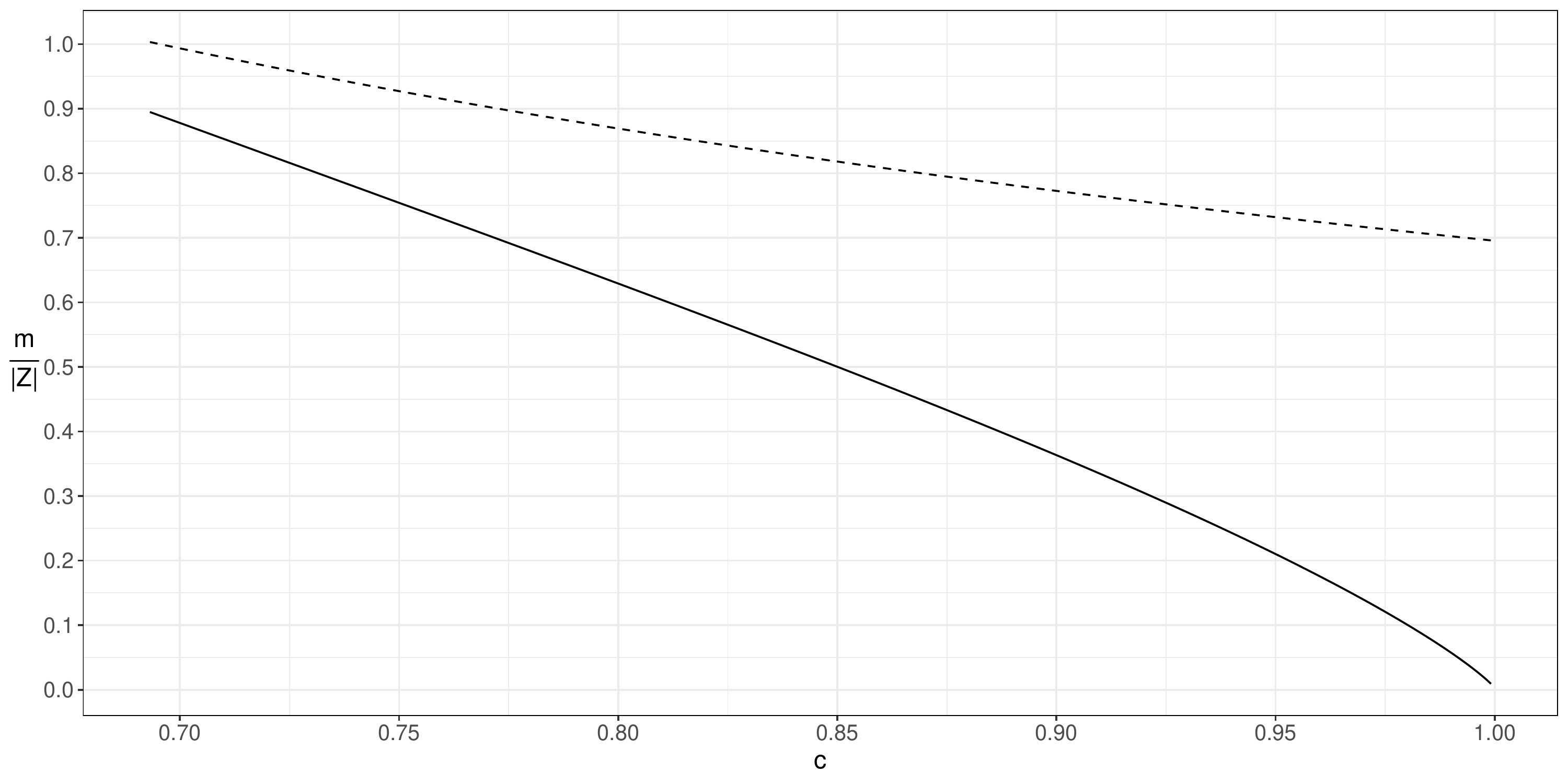}
	\caption{Sample size $m_{\mathcal{H}}\left(c,\delta\right)$ given by \eqref{psample_size} (dashed) and its upper bound \eqref{bound_m} when $q$ conditioned on the good parameters is uniform (solid line), as a function of $c$, with $\delta = 0.1$, $G(\{X_{\theta^{\star}}\}_{t = 1}^{T}) = 0.001,$ $|\mathcal{Z}| = 10^{6}$ and $p = 1,000$.} \label{m}
\end{figure}

Since the sample size above is distribution-free, i.e., true for any distribution, it will be an overestimated number in most cases. Even the sample size when $q$ conditioned on the good parameters is uniform tends to be overestimated. For example, if $q$ is the uniform distribution on $\mathcal{Z}$, in order to obtain a proportion $1-c$ of the good parameters one should have a sample of the order $(1-c)|\mathcal{Z}|$ which is lesser than the one calculated above. However, when $q$ is far from the uniform distribution, the researcher does not have control over $q$, i.e., he cannot choose it, or it is unknown, the sample size above may be a tool to assess beforehand the computational complexity of the estimation process, rather than a number that must be achieved. 

\subsection{Relation to ABC}
\label{relABC}

The method developed above is analogous to ABC rejection sampler when the prior distribution $\pi(\theta)$ has a discrete domain and $f(x|\theta)$ is deterministic. Observe that the rejection test is a special case of a fitness map and the good parameters $\hat{\Theta}$ is a sample of the approximation $\pi(\theta|F(\theta,\{X_{\theta^{\star}}\}_{t = 1}^{T}) = 1)$ of the posterior distribution. Hence, the method is a special case of ABC: a reinterpretation of it through the lens of Statistical Learning.

For practical purposes of obtaining qualitative properties of the system evolution, as will be exemplified below, more optimized ABC methods, such as ABC MCMC or ABC SMC, may be applied to calculate $\hat{\Theta}$ since they return a sample of \textit{good parameters} analogous to the rejection sampler. Furthermore, one may consider a distribution $\pi(\theta)$ with infinite domain to carry out the qualitative analysis, even though the method and sample size developed above does not cover this case. Indeed, in the qualitative analysis we will interpret the behaviour of the estimated parameters of a disease spread over time so, at principle, they could be estimated by any ABC method. However, in next section, we apply the method as described above to show how it works in practice.

\section{Examples}
\label{Sec4}

\subsection{SIR}

The spread of infectious diseases within a population is generally studied via compartmental epidemiological models \cite{capasso2008,vynnycky2010}, in which an individual is in either one of a handful of states related to the disease (Susceptible, Exposed, Infected, Recovered, etc...), and changes states according to some rates. In these models, the rates, which have epidemiological meaning, are the parameters, which should be estimated for each disease and different populations. 

We consider the family of Susceptible-Infected-Recovered (SIR) dynamical systems generated by the following difference equations
\begin{linenomath}
	\begin{equation}
	\label{SIR}
	\begin{cases}
	S(t+1) - S(t) = - \beta \ I(t)S(t)\\
	I(t+1) - I(t) = \beta \ I(t)S(t) - \gamma \ I(t)\\
	R(t+1) - R(t) = \gamma \ I(t)
	\end{cases}
	\end{equation}
\end{linenomath}
with $(S(0),I(0),R(0)) = c(0.95,0.05,0)$. In this case, $\Omega = \{x \in [0,1]^{3}: x_{1} + x_{2} + x_{3} = 1\}$, $\theta = (\beta,\gamma)$ and $\Theta = \mathbb{R}_{+}^{2}$. The observed evolution will be that generated by $\theta^{\star} = (0.25,1/21)$, for $t \leq 10$, and the fitness map will be
\begin{linenomath}
	\begin{equation*}
	F(\theta,\{X_{\theta^{\star}}(t)\}_{t=1}^{10}) = \min\limits_{t \in \{1,\dots,T\}} \min\limits_{i \in \{1,2,3\}}  \mathds{1}\Big\{|(X_{\theta}(t))_{i} - (X_{\theta^{\star}}(t))_{i}| \leq r \ (X_{\theta^{\star}}(t))_{i} \Big\}\\
	\end{equation*}
\end{linenomath}
so a parameter $\theta$ is good if the evolution $\{X_{\theta}(t)\}_{t=1}^{10}$ is within $100r\%$ of $\{X_{\theta}(t)\}_{t=1}^{10}$ for all three compartments in all times $t \leq T$, for $0 < r < 1$.

In the simulations, we consider as candidate models a grid of squares with side $0.001$ of the following three subsets of $\Theta$
\begin{linenomath}
	\begin{align*}
	\mathcal{Z}_{1} = (0,1] \times (0,0.5] & & \mathcal{Z}_{1} = (0,1] \times (0,0.2] & & \mathcal{Z}_{3} = [0.1,0.5] \times (0,0.2]
	\end{align*}
\end{linenomath}
which represent distinct prior knowledge about parameters $\beta$ and $\gamma$. If one does not known much about the parameters, he may assume that $\beta$ is lesser than $1$ and $\gamma$ lesser than $0.5$, considering $\mathcal{Z}_{1}$. If one has prior information about $\gamma$, that it is no more than $0.2$, then he may consider $\mathcal{Z}_{2}$. Finally, if one has also prior knowledge about $\beta$, that it is between $0.1$ and $0.5$, then he may consider $\mathcal{Z}_{3}$. The cardinality of these sets are, respectively, $5 \times 10^{5}$, $2 \times 10^{5}$ and $80,200$. We take $q(\theta)$ as the uniform measure on the respective set $\mathcal{Z}_{i}$.

Two values of $r$, namely, $0.05$ and $0.1$, are considered, and $c = 0.9$ is chosen to obtain $10\%$ of the good models with a confidence of $99\%$ ($\delta = 0.01$). To illustrate the method, an exhaustive search of all sets was performed in order to calculate $G(\{X_{\theta^{\star}}\}_{t = 1}^{T})$ and $p$, even though it will not be performed when applying the method on practice, and the results are on Table \ref{table_sirsize}. The value of $p$ is the same in all sets, although $G(\{X_{\theta^{\star}}\}_{t = 1}^{T})$ is greater in the smallest set $\mathcal{Z}_{3}$, that is the one which incorporated more prior information about the problem. Also, an upper bound for the sample size needed to obtain a same proportion of these good parameters, 10\% in this case, is much lesser in the smallest set, illustrating that as more prior information about the problem at hand is incorporated in the estimation process, less samples are needed. Observe that the fitness map with $r = 0.1$ is a more loose definition of goodness-of-fit, hence there are almost four times the number of good models when $r = 0.05$.

\begin{table}[ht]
	\centering
	\caption{Values of $G(\{X_{\theta^{\star}}\}_{t = 1}^{T})$, $p$ and $m(0.9,0.01)$ for each set and fitness map.} \label{table_sirsize}
	\begin{tabular}{ccccc}
		\hline
		$r$ & $\mathcal{Z}$ & $G(\{X_{\theta^{\star}}\}_{t = 1}^{T})$ & $p$ & $m(0.9,0.01)$ \\ 
		\hline
		\multirow{3}{*}{0.05} & $\mathcal{Z}_{1}$ & 0.000136 & 68 & 211,219 \\ 
		& $\mathcal{Z}_{2}$ & 0.00034 & 68 & 84,487 \\ 
		& $\mathcal{Z}_{3}$ & 0.00084 & 68 & 33,879 \\ 
		\hline
		\multirow{3}{*}{0.10} & $\mathcal{Z}_{1}$ & 0.000526 & 263 & 186,534 \\ 
		& $\mathcal{Z}_{2}$ & 0.001315 & 263 & 74,613 \\ 
		& $\mathcal{Z}_{3}$ & 0.003279 & 263 & 29,920 \\ 
		\hline
	\end{tabular}
\end{table}

Figure \ref{SIRfig} shows the evolution of the good models obtained in each scenario, and Table \ref{SIRtable} presents descriptive statistics of the good parameters and the peak of infected, i.e., the day with more simultaneously infected individuals, simulated by the good parameters. The evolution generated by the good parameters is close to the evolution of the SIR for all times, not only the first ten which were used in the estimation. Also, Table \ref{SIRtable} shows that the values of $\beta$ and $\gamma$ of the good parameters are close to the real ones. Finally, the peak of infected individuals, which for the observed model occurs at day $24$, was predicted within an error of no more than two days by the evolution generated by the good parameters in all scenarios. We conclude that, by applying the method developed in this paper, we may obtain information about the parameters and qualitative behaviour of the SIR model even when we do not have much prior information about the parameters ($\mathcal{Z}_{1}$) or when we choose a more loose fitness map ($r = 0.1$).

\begin{table}[ht]
	\centering
	\caption{Sample size (N), number of good parameters in the sample ($|\hat{\Theta}|$) and descriptive statistics of $\beta$, $\gamma$ and the peak of the disease for the evolution generated by the good parameters in each scenario. The true values of $\beta$, $\gamma$ and the peak are $0.25, 0.0476$ and $24$, respectively.} \label{SIRtable}
	\resizebox{\linewidth}{!}{
		\begin{tabular}{|cccc|cccc|cccc|cccc|}
			\hline
			\multirow{2}{*}{$\mathcal{Z}$} & \multirow{2}{*}{r} & \multirow{2}{*}{N} & \multirow{2}{*}{$|\hat{\Theta}|$} & \multicolumn{4}{c|}{$\beta$} & \multicolumn{4}{c|}{$\gamma$} & \multicolumn{4}{c|}{Peak} \\ \cline{5-16}
			& & & & Mean & Median & Min & Max & Mean & Median & Min & Max & Mean & Median & Min & Max \\ 
			\hline
			\multirow{2}{*}{$Z_{1}$} & 0.05 & 211,219 &   25 & 0.249 & 0.248 & 0.242 & 0.256 & 0.047 & 0.047 & 0.046 & 0.049 & 24 & 24 &   23 &   25 \\ 
			& 0.1 & 186,534 &   86 & 0.249 & 0.248 & 0.232 & 0.266 & 0.048 & 0.048 & 0.043 & 0.052 & 23.988 & 24 &   23 &   26 \\ 
			\hline
			\multirow{2}{*}{$Z_{2}$} & 0.05 & 84,488 &   30 & 0.250 & 0.249 & 0.241 & 0.258 & 0.047 & 0.047 & 0.046 & 0.049 & 23.933 & 24 &   23 &   25 \\ 
			& 0.1 & 74,614 &   85 & 0.251 & 0.251 & 0.232 & 0.268 & 0.047 & 0.047 & 0.043 & 0.052 & 23.918 & 24 &   22 &   26 \\ 
			\hline
			\multirow{2}{*}{$Z_{3}$} & 0.05 & 33,880 &   25 & 0.250 & 0.249 & 0.242 & 0.257 & 0.048 & 0.048 & 0.046 & 0.049 & 23.880 & 24 &   23 &   25 \\ 
			& 0.1 & 29,920 &   90 & 0.250 & 0.251 & 0.232 & 0.268 & 0.047 & 0.048 & 0.043 & 0.052 & 23.978 & 24 &   22 &   26 \\    
			\hline
	\end{tabular}}
\end{table}

\begin{figure}[ht]
	\centering
	\includegraphics[width = \linewidth]{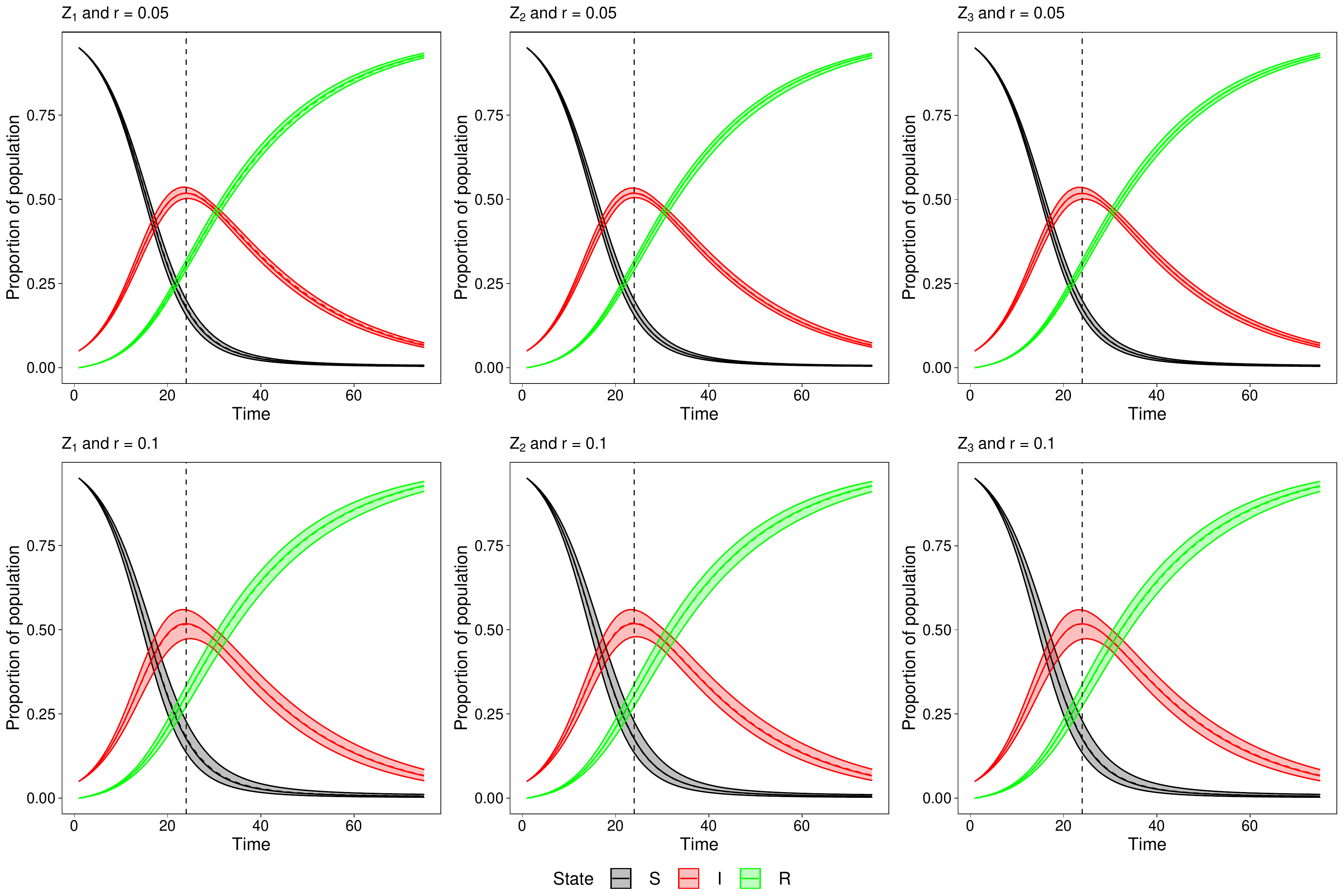}
	\caption{The SIR evolution (solid line), and the mean evolution generated by the good parameters (dashed line) in each scenario. The ribbon refers to the area between the minimum and the maximum evolution generated by the good parameters and the dashed vertical line is the peak of infected.} \label{SIRfig}
\end{figure}

\subsection{SEIR for COVID-19}

In order to mitigate the effects of COVID-19 pandemic, the greatest world health crisis in a century, governments and the population have sought reliable information about its spread. More than quantitative predictions about cases and deaths tolls, the main questions regarding the disease spread are about certain milestones such as its peak, when it will be over, if the measures implemented to contain it are effective and if or when the available intensive care units will be full. In this section we illustrate, with data about its spread in the US, how the method developed in this paper may be useful in answering such questions about COVID-19.

In order to account for features of COVID-19, we add a compartment to the usual SEIR model (see \cite{hethcote2000} for a review of compartmental models). After being exposed to the disease, an individual becomes Infected with rate $\gamma_{I}$ and is able to infect others. However, once infected, an individual either recovers with rate $\nu_{R}$ or is accounted by the official statistics becoming Infected in Statistics with rate $\gamma_{S}$, not being able to infect others any more. 

Hence, we are assuming that a parcel of the infected is never accounted by official statistics and that, once diagnosed with COVID-19, an individual will not infect others, as he will be either admitted to a hospital or be in quarantine. These assumptions account for features of the disease such as the under-notification of cases. We also add a compartment for Deaths by the disease and assume that natural deaths and births balance out, so they are not considered to simplify the model, which is illustrated in the diagram of Figure \ref{diagram}.

\begin{figure}[H]
	\centering
	\begin{tikzpicture}[scale=0.5, transform shape]
	\tikzstyle{hs} = [draw, rounded corners, minimum height=4em, minimum width=5em, node distance=1.5cm, line width=1pt, scale=1.5]

	\node at (-6,0) [hs] (S) {Susceptible};
	\node at (0,0) [hs] (E) {Exposed};
	\node at (6,-5) [hs] (I^{S}) {Infected in Statistics};
	\node at (6,0) [hs] (I) {Infected};
	\node at (12.5,0) [hs] (R) {Recovered};
	\node at (12.5,-5) [hs] (D) {Death};
	
	\draw[->] (S) -- (E) node[pos=0.5,sloped,above,scale=1.5] {$\beta$};
	\draw[->] (E) -- (I) node[pos=0.5,sloped,above,scale=1.5] {$\gamma_{I}$};
	\draw[->] (I) -- (I^{S}) node[pos=0.5,sloped,right,scale=1.5,rotate = 90] {$\gamma_{S}$};
	\draw[->] (I) -- (R) node[pos=0.5,sloped,above,scale=1.5] {$\nu_{R}$};
	\draw[->] (I^{S}) -- (R) node[pos=0.5,sloped,below,scale=1.5] {$\nu_{RS}$};
	\draw[->] (I^{S}) -- (D) node[pos=0.5,sloped,below,scale=1.5] {$\delta$};
	\end{tikzpicture}
	\caption{Diagram of the SEIR model for COVID-19.}
	\label{diagram}
\end{figure}
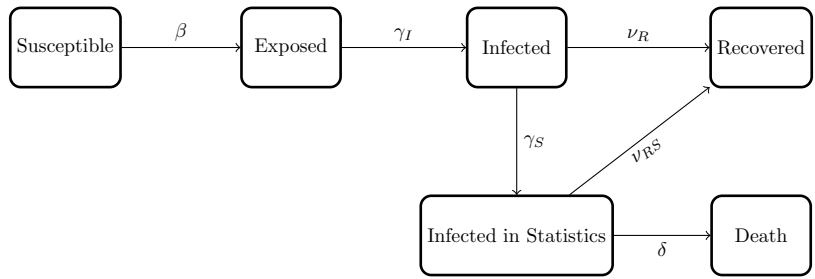

The evolution of the Susceptible ($S$), Exposed ($E$), Infected ($I$), Infected in Statistics ($I_{s}$), Recovered ($R$) and Deaths ($D$) by COVID-19 is ruled by the difference equations
\begin{linenomath}
	\begin{align*}
	S(t+1) - S(t) &= - \frac{\beta S(t)}{(N - D(t))} I(t) \\
	E(t+1) - E(t) &= - \gamma_{I} E(t) + \frac{\beta S(t)}{(N - D(t))} I(t) \\
	I(t+1) - I(t) &=- \left(\nu_{R} + \gamma_{S}\right) I(t) + \gamma_{I} E(t) \\
	I_{s}(t+1)  - I_{s}(t) &= - (\nu_{RS} + \delta)I_{s}(t) + \gamma_{S} I(t) \\
	R(t+1) - R(t) &= \nu_{R} I(t) + \nu_{RS} I_{s}(t) \\
	D(t+1) -  D(t) &= \delta I_{s}(t)
	\end{align*}
\end{linenomath}
in which $N$ is the population size. We assume that the time scale is that of days, so $t+1$ is the day after $t$, and consider the absolute number of individuals in each compartment, rather than the proportion of the population.

This model has a nice biological meaning and its parameters are related to characteristics of the disease. Indeed, apart from the force of infection $\beta$, they are function of the disease average exposed time $\tau_{E}$, average time infected before getting into statistics $\tau_{S}$, average time infected before recovering without being accounted by statistics $\tau_{R}$, average time before recovering after being accounted by statistics $\tau_{RS}$, average time until death once accounted by statistics $\tau_{D}$, and the proportions of infected individuals accounted by statistics $p_{S}$, and of deaths among them $p_{D}$:
\begin{linenomath}
	\begin{align*}
	\gamma_{I} = \frac{1}{\tau_{E}} & & \gamma_{S} = \frac{p_{S}}{\tau_{S}} & & \nu_{R} = \frac{(1-p_{S})}{\tau_{R}} & & \nu_{RS} = \frac{(1-p_{D})}{\tau_{RS}} & & \delta = \frac{p_{D}}{\tau_{D}}.
	\end{align*}
\end{linenomath}
Hence, after being exposed, an individual becomes infected after a time with mean $\tau_{E}$. Among the infected, a proportion $p_{S}$ will be accounted by official statistics after a time with mean $\tau_{S}$, and will not be able to infect any more, while a proportion $1-p_{S}$ will recover after a time with mean $\tau_{R}$ without ever being accounted by official statistics. Then, a proportion $p_{D}$ of the Infected in Statistics will die after a time with mean $\tau_{D}$ and a proportion $1-p_{D}$ will recover after a time with mean $\tau_{RS}$. Once recovered, an individual cannot become infected again.

Apart from $p_{D}$, which can be estimated directly from data as the proportion of deaths among infected accounted by statistics, the other parameters need to be estimated. Since the rates of the model can be calculated from the mean average time in each compartment and the proportion of infected accounted by statistics, we consider
\begin{linenomath}
	\begin{equation*}
	\theta \coloneqq \left(\beta,\tau_{E},\tau_{R},\tau_{S},\tau_{RS},\tau_{D},p_{S}\right) \in \Theta = \mathbb{R}_{+}^{6} \times [0,1]
	\end{equation*}
\end{linenomath}
as the free parameters to be estimated. The candidate values for each parameter are presented in Table \ref{tCandidate}. The number of candidate models is
\begin{linenomath}
	\begin{equation*}
	|\mathcal{Z}| = 116,121,600
	\end{equation*}
\end{linenomath}
since $\mathcal{Z}$ is a grid of $\Theta$ generated by the values in Table \ref{tCandidate}.

\begin{table}
	\centering
	\caption{Candidate parameters of the model} \label{tCandidate}
	\resizebox{\linewidth}{!}{\begin{tabular}{cc}
			\hline
			Parameter & Candidates \\
			\hline
			$\beta$ & $0.05,0.1,0.15,0.2,0.25,0.3,0.35,0.4,0.45,0.5,0.6,0.7,0.8,0.9,1,1.1,1.2,1.3,1.4,1.5$\\
			$\tau_{E}$ & $4,5,6,7$ \\
			$\tau_{R}$ & $5,6,7,\dots,14$\\
			$\tau_{S}$ & $3,4,5,\dots,14$\\
			$\tau_{RS}$ & $5,6,7,\dots,28$\\
			$\tau_{D}$ & $1,2,3,\dots,28$\\
			$p_{S}$ & $0.01,0.025,0.05,0.075,0.1,0.15,0.2,0.25,0.3,0.35,0.4,0.5,0.6,0.7,0.8,0.9,0.95,0.99$\\
			\hline
	\end{tabular}}
\end{table}

In order to illustrate the method, we estimate the parameters multiple times considering goodness-of-fit during consecutive periods of 7 days starting on March 20th 2020, i.e., we apply the method for multiple $t_{0}$, with 7 days apart from one to the next, starting on March 20th until mid-July. Goodness-of-fit is accomplished by testing the predicted evolution against data about COVID-19 spread in the US compiled by Johns Hopkins Coronavirus Resource Center \cite{dong2020}. The observed data is smoothed by taking a centred seven day average of the incidence of confirmed cases, deaths and recovered, and the prevalences calculated by summing the smooth incidences. 

Since only two compartments, namely Infected in Statistics and Deaths, are fully observed, the goodness-of-fit of a candidate model should be defined comparing its predicted with the observed values in these compartments. Denote $X(t) = (I_{s}(t),D(t))$ and $X_{\theta}(t) = (I_{s}(t;\theta),D(t;\theta))$ the observed and predicted by model with parameter $\theta$, respectively, Infected in Statistics and Deaths at time $t$. With this notation, the fitness map is
\begin{linenomath}
	\begin{equation*}
	F(\theta,\{X(t)\}_{t=t_{0}}^{t_{0}+6}) = \min\limits_{t \in \{t_{0},\dots,t_{0}+6\}}  \mathds{1}\left\{\left\lVert \frac{X_{\theta}(t) - X(t)}{X(t)} \right\rVert_{\infty} \leq r(t_{0})\right\}
	\end{equation*}
\end{linenomath}
for each $t_{0}$ considered and $0 < r(t_{0}) < 1$, in which division by $X(t)$ is component-wise. Hence, a parameter $\theta$ well fit the observed data if the predicted value of Infected in Statistics and Deaths is within $100r(t_{0})\%$ of the respective observed value for everyday in the week starting on $t_{0}$. The value of $r$ must depend on $t_{0}$ since there are moments when the disease spread is better explained by a SEIR model; for example when few measures are implemented to slow its evolution, as will be seen below. 

We estimated the model every seven days in order to asses the peak of deaths if the disease were to keep spreading as in the week starting on $t_{0}$. Due to the enforcement of measures to flatten the epidemiological curve, the parameters may actually be time-dependent, specially the force of infection $\beta$. Hence, estimating the peak assessing goodness-of-fit only against the first weeks of the disease spread is doomed to fail, since the force of infection changes from time to time. 

On the other hand, by fitting the model every week, one can see changes on the peak estimate over time, which will reflect the enforcement or looseness of measures to slow the spread, providing a robust and better picture of the future. Observe that the method provides evidence about the disease spread if it were to continue as in the week when goodness-of-fit was assessed, not providing forecasts of distinct scenarios.

To estimate the parameters and then simulate the good models starting on $t_{0}$, we need to estimate the initial condition for compartments $E$ and $I$, which are not observed, and $R$ which is only partially observed. Assume that we fixed a vector $\theta$ of candidate parameters. Denoting $R_{s}$ as the recovered among the accounted by statistics, we take
\begin{linenomath}
	\begin{align*}
	I(t) = \frac{(1-p_{S})}{p_{S}} I_{s}(t) & & R(t) =  \left[\frac{(1-p_{S})}{p_{S}} + 1\right] R_{s}(t).
	\end{align*}
\end{linenomath}
Proceeding in this manner, we have that among all infected ($I + I_{s}$) and recovered ($R$), a proportion $p_{S}$ is accounted by statistics.

In order to estimate the exposed, we obtain from the equations ruling the evolution of the disease that
\begin{linenomath}
	\begin{align*}
	I(t+1) = - \left(\nu_{R} + \gamma_{S} - 1\right) I(t) + \gamma_{I} E(t)
	\end{align*}
\end{linenomath}
hence
\begin{linenomath}
	\begin{align*}
	E(t) &= \frac{1}{\gamma_{I}} \left[I(t+1) + (\nu_{R} + \gamma_{S} - 1) I(t)\right]
	\end{align*}
\end{linenomath}
so we calculate $E(t_{0})$ from $I(t_{0})$ and $I(t_{0}+1)$. We take $p_{D}$ as the moving death rate considering only the confirmed cases and deaths according to the smoothed data on the week ending on $t_{0}$, so it expresses the death rate around $t_{0}$. Observe that each candidate model has a distinct initial condition on the unobserved compartments, dependent on its parameters. 

The error $r(t_{0})$ was determined by sampling $100,000$ models and observing the minimum difference 
\begin{linenomath}
	\begin{equation*}
	\min\limits_{t \in \{t_{0},\dots,t_{0}+6\}} \left\lVert \frac{X_{\theta}(t) - X(t)}{X(t)} \right\rVert_{\infty}
	\end{equation*}
\end{linenomath}
over the sampled $\theta$, which gives an estimative of \textit{how well} the SEIR model approximates the evolution on the week starting on $t_{0}$. The error $r(t_{0})$ was then taken as $1.1$ times this minimum difference. We could not choose an absolute $r$ beforehand since there is not a model approximating \textit{arbitrarily well} the evolution for all $t_{0}$ as it does not follow exactly a SEIR model. 

The values of $r(t_{0})$ were around $0.03$ until April 10th when it dropped to around $0.01$. On May and until mid-June the values of $r(t_{0})$ were more or less 0.07, dropped to $0.03$ on the second half of June and then to $0.01$ on the first half of July. Fixed $r(t_{0})$, for each $t_{0}$ there were sampled $500,000$ models.

Figure \ref{par_US} shows the box-plot of the parameters of the good models for each $t_{0}$, and the number of good models ($m$) among the sampled ones. Figure \ref{curve_US} shows the Infected in Statistics and Deaths according to the smoothed data and simulated by the good models for the seven days starting on their $t_{0}$.

We first see that the distribution of $\tau_{E}, \tau_{R}$ and $\tau_{RS}$, which are mainly related to features of the disease, does not vary very much over time. Also, we see that the good models better fit observed data when implemented measures to slow the spread were not effective, namely, until the end of April, when it was spreading mainly in the Northeast region, and from mid-June when it started spreading rapidly in other states, specially in the South. This is reasonable, since SEIR models assume that every individual which has not ever been infected is susceptible to the disease and that any contact between an infected and a susceptible has a same probability of infection, which is not true if part of the population is on lockdown or if social distance and other measures, such as mask wearing, are implemented. Hence, it is expected a poor performance of the model when such measures are implemented, at least in what regards the number of infected and deaths.

\begin{figure}[ht]
	\centering
	\includegraphics[width=\linewidth]{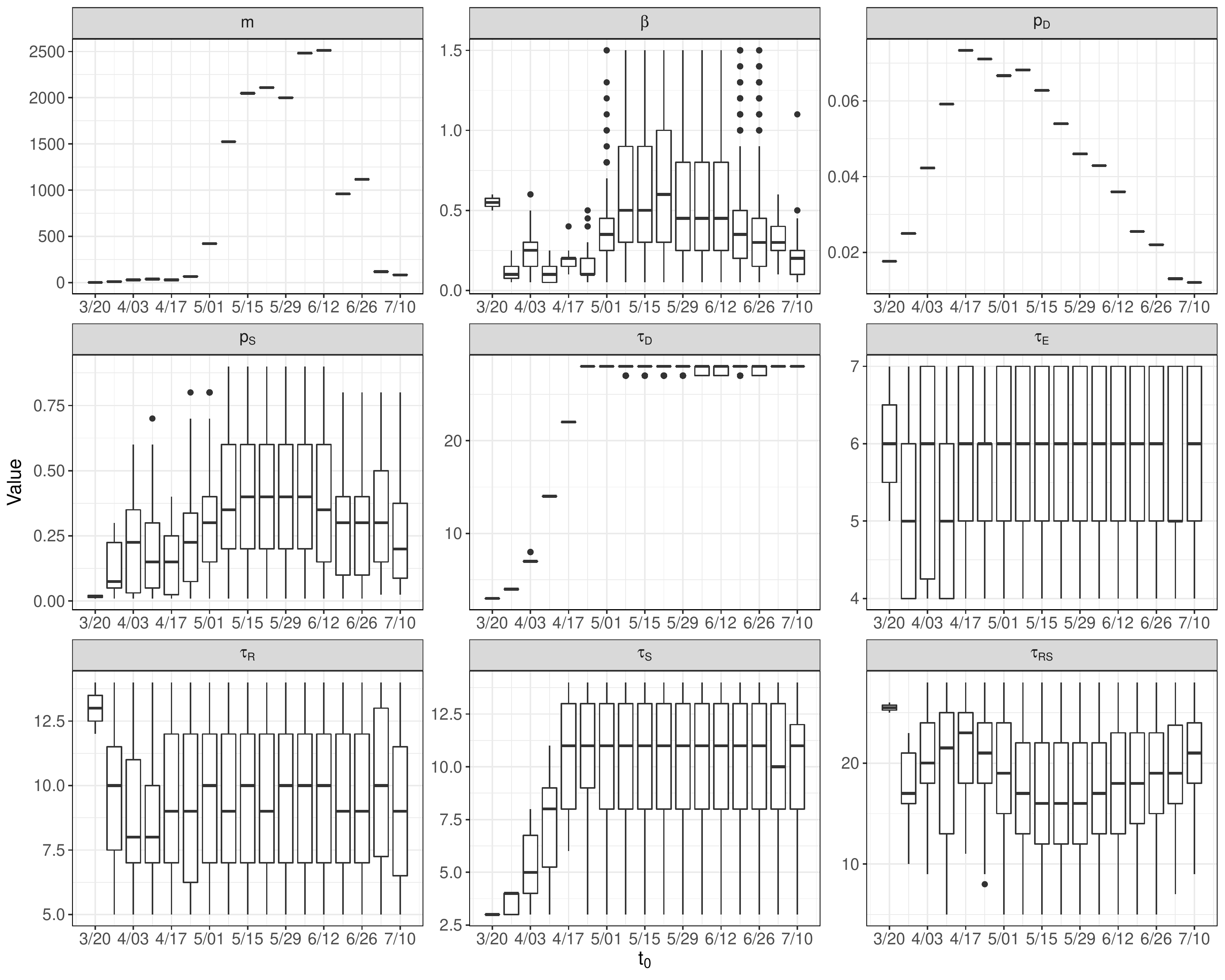}
	\caption{Number of good parameters among the $500,000$ sampled ($m$) and box-plot of the set of good parameters sampled for each $t_{0}$. The value of $p_{D}$ is calculated directly from data hence is the same for all models.} \label{par_US}
\end{figure}

\begin{figure}[ht]
	\centering
	\includegraphics[width=\linewidth]{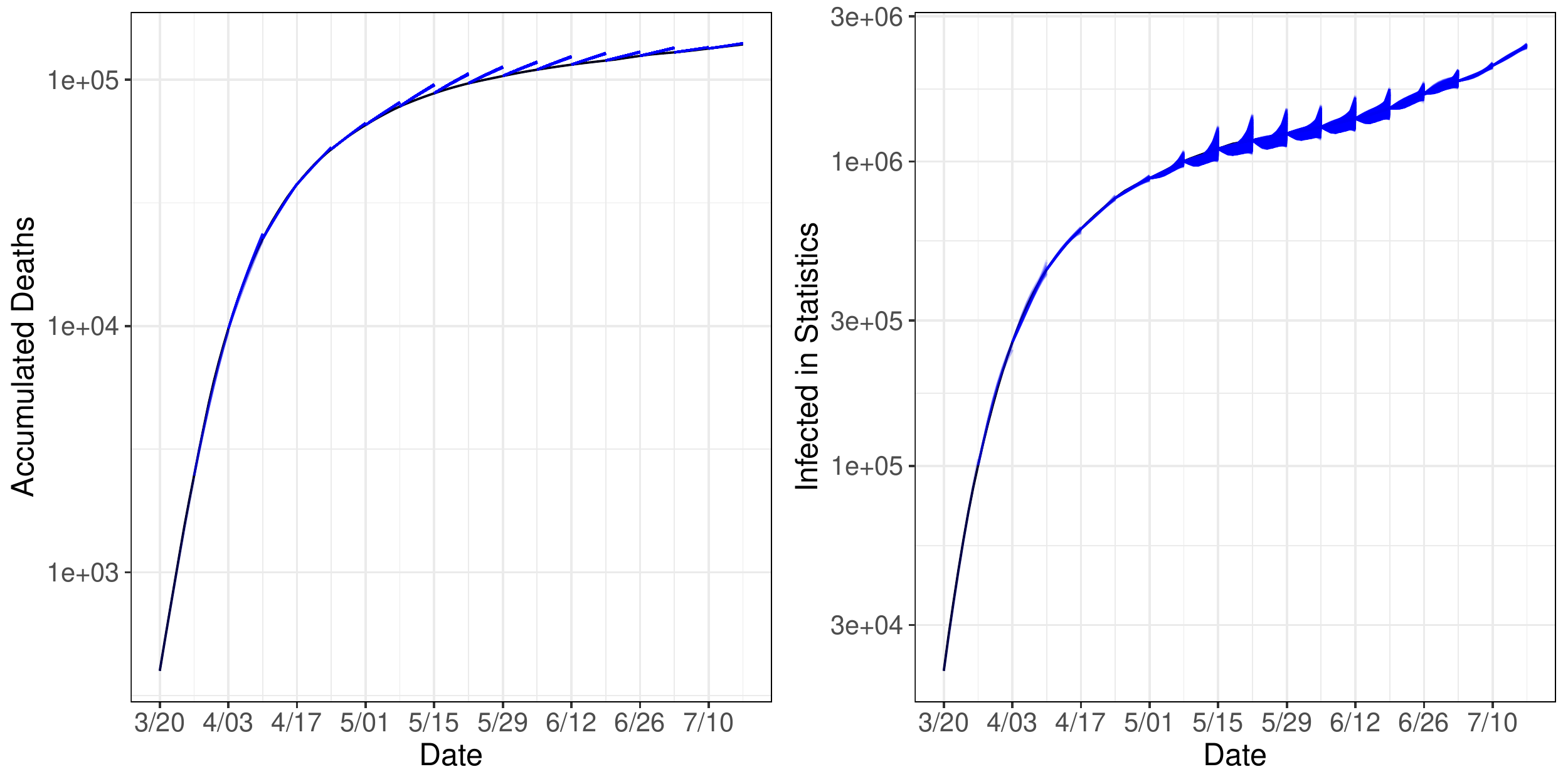}
	\caption{Infected in Statistics and Deaths according to the smoothed data (solid line) and the good models (blue lines). The good models were simulated only on the seven days starting on their $t_{0}$.} \label{curve_US}
\end{figure}

The distribution of the estimated parameters over time evidences interesting qualitative features of the disease spread in the US. The moving rate of death $p_{D}$ started low, but more than tripled in the following weeks, then slowly decreased over time, what may have been due to the increase on testing or the fact that deaths among the new infected is yet to occur \cite{times1}. 

The number of good models was greater when the disease was spreading more slowly, what is due to the employment of a more loose fitness map, since a SEIR model could not \textit{well fit} data, in a absolute sense, in this period, as can be seen in the simulated and observed curves of Figure \ref{curve_US}.

The values of $\tau_{S}$ and $\tau_{D}$ started small and then increased until an equilibrium around $11$ and $27$ days, respectively. This may be due to the fact that early on there was a lot of people dying at home, specially on nursing homes \cite{times2}, so they died rapidly due to underlying health problems, not having the time to procure health assistance, so their time infected before getting into statistics was small, and they were accounted as infected in statistics and deaths at the same time, so their time until death after being accounted by statistics was zero; these cases cause the mean times $\tau_{S}$ and $\tau_{D}$ to be smaller. 

Observe that, when the model was not absolutely well fitting data, the range of most estimated parameters contemplated all the candidates, what tells us that, with a right combination of parameters, any given value of a parameter can generate an evolution that fit the data, according to the loose definition of well fit. On the other hand, when the model was absolutely well fitting the data, until the end of April and from mid-June, the parameters are more meaningful. This is specially true for $p_{S}$, which started low until mid-April, and when the disease spread increased once more at mid-June it was again low, but greater than on March and April, illustrating the increase on testing. 

Finally, the parameter $\beta$ tended to be greater on mid-June than on the time until the end of April, evidencing that the force of infection in the second moment of fast spread is greater; the values when the model was not fitting data are not meaningful since vary over the range of candidate parameters.

Figure \ref{peak_US} presents the median and selected percentiles of the peak of deaths, i.e., the day with more deaths after $t_{0}$, estimated by the good models in each $t_{0}$. The models fitted between March 20th and April 3rd predicted that the the peak could happen around April 15th, which was the peak observed on the smoothed data, since the percentile 2.5\% of them was on this day. Also, for $t_{0}$ equal to April 10th, the median of the models was pointing to a close peak, although it got it wrong by one week. After the first peak, for some $t_{0}$ it was predicted as a good scenario no more peaks (when the peak is in $t_{0}$), but then, for the months in which the disease was spreading more slowly, the peak was predicted for a sequence of $t_{0}$ as around one month after it. This was the behaviour until the end of June when the peak started to be predicted as more distant, evidencing the future increase on the incidence of deaths in the mid-term before it starts decreasing.

\begin{figure}[ht]
	\centering
	\includegraphics[width=\linewidth]{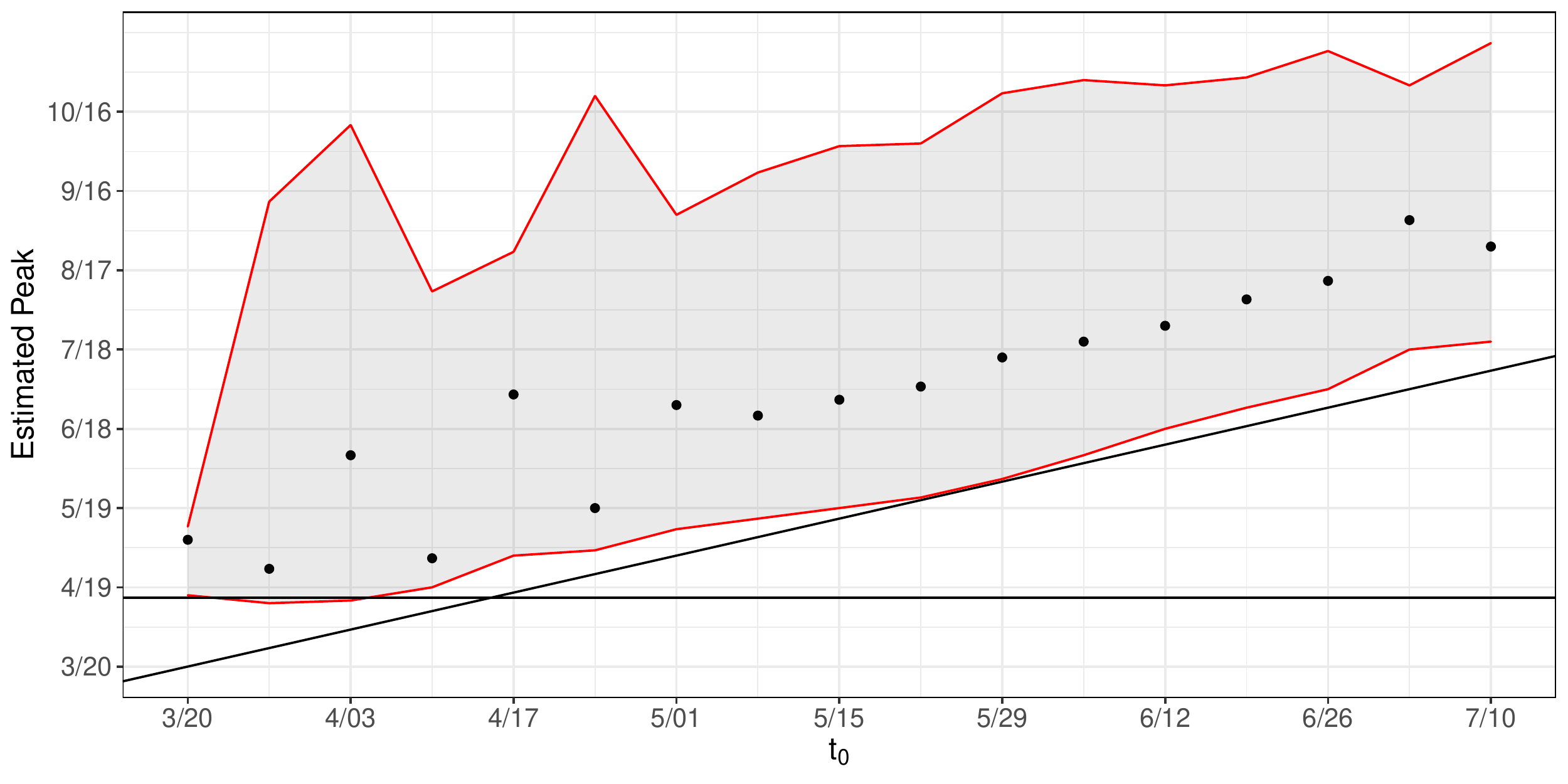}
	\caption{Median estimate for the peak of deaths among the good models sampled for each $t_{0}$ (dots), with the percentiles 2.5 and 97.5\%, represented in red. The horizontal line represents the first peak of deaths according to the smoothed data. The solid line is the identity function.} \label{peak_US}
\end{figure}

This example illustrates how such a simple estimation method may offer many insights into qualitative properties of the disease spread, aiding in assessing its current evolution and predicting its future spread. By fitting this model every week during an outbreak, one may analyse its predictions and estimated parameters over time to assess the effectiveness of measures implemented to slow the spread and some qualitative features of the disease current and future evolution. 

\section{Final Remarks}
\label{Sec5}

The estimation method for dynamical systems parameters proposed in this paper aims to find a subset $\hat{\Theta} \subseteq \Theta$ of parameters that well fit the observed evolution in a time interval. This is a more humble agenda when compared with methods that try to find exactly the parameter which generated the evolution, and is quite robust when evidencing qualitative behaviour of evolutions, specially ones that are only approximated by dynamical systems, as was illustrated by the study of the COVID-19 spread on the US.

In essence, the proposed method is a mathematical formalization of the procedure consisting of trial and error for parameter estimation, in which the evolution generated by candidate parameters are tested against observed data to determine which well fit it. The Statistical Learning framework introduces a mathematically rigorous scheme to this general approach for parameter estimation, reinterpreting an ABC method.

The method is quite general, as it does not rely on any specific property of the system, and highly flexible, since one may define the set of candidate parameters $\mathcal{Z}$, the fitness map $F$ and select any sampling algorithm on $\mathcal{Z}$ to obtain distinct methods for parameter estimation. An interesting topic for future research is to consider $\mathcal{Z}$ with infinity cardinality, but choose $\mathcal{H}$ and $F$ in a way such that $\mathcal{H}$ has finite VC-dimension, so results analogous to Theorem \ref{bound_error} are true. With such a result, one could consider sampling parameters from a set with infinity cardinality relying on results of Statistical Learning.

The method may be quite useful for disease spread models which have the properties it was developed to address: disease spread does not exactly follow a compartmental model and there is usually an interest in the qualitative behaviour of the evolution, rather than in its exact trajectory. These qualitative properties may aid the population and government officials in the decision making process during an outbreak, hence, as illustrated with the COVID-19 spread in the US, fitting a compartmental model weekly may be a rich source of information about the disease spread, even if it cannot give reliable predictions on the number of cases and deaths for the long-term.

The shortcomings of the method are the need for high computational power to simulate the evolution generated by thousands of parameters, and the need to carefully define $\mathcal{Z}$ beforehand, properly identifying the set of candidate parameters. Nevertheless, it may be a good option to robustly estimate parameters in complex models, specially when there is an interest in the qualitative behaviour of the trajectories.

\section*{Acknowledgements}

I thank C. Peixoto, P. Peixoto and S. Oliva for fruitful discussions about the modelling of disease spread, specially COVID-19, and about the development of compartmental models to address it. The author has received financial support from CNPq during the development of this paper.

\bibliographystyle{vancouver}      
\bibliography{ref}   

\appendix
\section{Proof of results}

\begin{proof}[Proof of Theorem \ref{bound_error}]
	Observe that $L_{\mathcal{D}_{N}}(\hat{h}) = 0$, since at least $L_{\mathcal{D}_{N}}(h^{\star}) = 0$. Now, if $L(\hat{h}) > \epsilon$, then 
	\begin{linenomath}
		\begin{equation*}
		\max\limits_{h \in \mathcal{H}: L_{\mathcal{D}_{N}}(h) = 0} L(h) > \epsilon
		\end{equation*}
	\end{linenomath}
	hence $\mathbb{P}\left(L(\hat{h}) > \epsilon\right)$ is lesser or equal to
	\begin{linenomath}
		\begin{equation*}
		\mathbb{P}\left(\max\limits_{h \in \mathcal{H}: L_{\mathcal{D}_{N}}(h) = 0} L(h) > \epsilon\right) = \mathbb{E}\left(\max\limits_{h \in \mathcal{H}} \mathds{1}\left\{L_{\mathcal{D}_{N}}(h) = 0\right\} \mathds{1}\left\{L(h) > \epsilon\right\}\right).
		\end{equation*}
	\end{linenomath}
	Applying the union bound on the last expectation above we obtain that
	\begin{linenomath}
		\begin{equation*}
		\mathbb{P}\left(L(\hat{h}) > \epsilon\right) \leq \sum_{h \in \mathcal{H}: L(h) > \epsilon} \mathbb{P}\left(L_{\mathcal{D}_{N}}(h) = 0\right) \leq |\mathcal{H}| (1- \epsilon)^{N}
		\end{equation*}
	\end{linenomath}
	since the probability that no pair $(Z_{i},Y_{i}), i = 1,\dots,N,$ falls in the set $\{(z,y): h(z) \neq y\}$ is less than $(1-\epsilon)^{N}$, as the probability of this set is $L(h) > \epsilon$ and the pairs are independent. The result follows since $(1-\epsilon) \leq e^{-\epsilon}$.
\end{proof}

\begin{proof}[Proof of Proposition \ref{prop_improve}]
	From the proof of Theorem \ref{bound_error} we have that
	\begin{linenomath}
		\begin{equation*}
		\mathbb{P}\left(L(\hat{h}) > cG(\{X_{\theta^{\star}}\}_{t = 1}^{T})\right) \leq \left|\{h \in \mathcal{H}: L(h) > cG(\{X_{\theta^{\star}}\}_{t = 1}^{T})\}\right| \ e^{-cG(\{X_{\theta^{\star}}\}_{t = 1}^{T}) N}.
		\end{equation*}
	\end{linenomath}
	Now, if $L(h) > cG(\{X_{\theta^{\star}}\}_{t = 1}^{T})$ then at least $cp$ of the good parameters are not in the subset generated by $h$ since measure $G(\{X_{\theta^{\star}}\}_{t = 1}^{T})$ is uniformly spread among the good parameters. As there are
	\begin{linenomath}
		\begin{equation*}
		\sum_{k=0}^{(1-c)p} \binom{p}{k}
		\end{equation*}
	\end{linenomath}
	subsets with no more than $(1-c)p$ good parameters the first assertion follows. To show the second assertion, by applying Corollary \ref{sample_size} we have
	\begin{linenomath}
		\begin{align}
		\label{eq_proof}
		m(c,\delta) &= \frac{1}{cG(\{X_{\theta^{\star}}\}_{t = 1}^{T})} \left[\log \sum_{k=0}^{(1-c)p} \binom{p}{k} - \log \delta\right].
		\end{align}
	\end{linenomath}
	Since the left-hand side of \eqref{eq_proof} is lesser than
	\begin{linenomath}
		\begin{equation*}
		\frac{1}{cG(\{X_{\theta^{\star}}\}_{t = 1}^{T})} \left[\log \binom{p}{(1-c)p} + \log\left(1 + \frac{(1-c)^{2} p^{2}}{cp + 1}\right) - \log \delta \right]
		\end{equation*}
	\end{linenomath}
	by Proposition \ref{bound_log}, we have the result.
\end{proof}

\section{Auxiliary results}

\begin{proposition}
	\label{bound_log}
	For $n \geq 2p$ integers
	\begin{linenomath}
		\begin{equation}
		\label{eq}
		\log \sum_{k=0}^{p} \binom{n}{k} \leq \log \binom{n}{p} + \log\left(1 + \frac{p^2}{n - p + 1}\right)
		\end{equation}
	\end{linenomath}
\end{proposition}
\begin{proof}
	By multiplying and dividing each term of the sum by $\binom{n}{p}$ we obtain that the expression in left-hand side of \eqref{eq} is equal to
	\begin{linenomath}
		\begin{equation*}
		\log \binom{n}{p} + \log \sum_{k=0}^{p} \frac{\binom{n}{k}}{\binom{n}{p}}.
		\end{equation*}
	\end{linenomath}
	Since for every $0 \leq k \leq p-1$
	\begin{linenomath}
		\begin{align*}
		\frac{\binom{n}{k}}{\binom{n}{p}} & \leq  \frac{\binom{n}{p-1}}{\binom{n}{p}} = \frac{p}{n - p + 1}
		\end{align*}
	\end{linenomath}
	the result follows from inequality
	\begin{linenomath}
		\begin{equation*}
		\sum_{k=0}^{p} \frac{\binom{n}{k}}{\binom{n}{p}} \leq 1 + \frac{p^2}{n - p + 1}.
		\end{equation*}
	\end{linenomath}
\end{proof}

\end{document}